\documentclass{article}

\usepackage{arxiv}

\usepackage[utf8]{inputenc} 
\usepackage[T1]{fontenc}    
\usepackage{hyperref}       
\usepackage{url}            
\usepackage{booktabs}       
\usepackage{amsfonts}       
\usepackage{nicefrac}       
\usepackage{amsmath}
\usepackage{cleveref}       
\usepackage{lipsum}         
\usepackage{graphicx}
\usepackage{doi}
\usepackage{csquotes}

\usepackage{natbib}

\usepackage{amsfonts}
\usepackage{amssymb}
\usepackage{amsthm}
\usepackage{verbatim}
\usepackage{fancyvrb}
\usepackage{longtable}
\usepackage{cprotect}
\usepackage{comment}
\usepackage{caption}
\usepackage{subcaption}
\usepackage{bm}

\title{Bayesian shrinkage priors subject to linear constraints}

\newif\ifuniqueAffiliation
\uniqueAffiliationtrue

\ifuniqueAffiliation 
\author{{Zhi Ling}\\
	Saw Swee Hock School of Public Health\\
	National University of Singapore\\
	\texttt{lingzhi@nus.edu.sg} \\
    \And
	{Shozen Dan} \\
	Department of Mathematics\\
	Imperial College London\\
	\texttt{shozen.dan21@imperial.ac.uk} \\
}
\else
\usepackage{authblk}

\newbox{\orcid}\sbox{\orcid}{\includegraphics[scale=0.06]{orcid.pdf}}
\author[1]{%
	\href{https://orcid.org/0000-0000-0000-0000}{\usebox{\orcid}\hspace{1mm}David S.~Hippocampus\thanks{\texttt{hippo@cs.cranberry-lemon.edu}}}%
}
\author[1,2]{%
	\href{https://orcid.org/0000-0000-0000-0000}{\usebox{\orcid}\hspace{1mm}Elias D.~Striatum\thanks{\texttt{stariate@ee.mount-sheikh.edu}}}%
}
\affil[1]{Department of Computer Science, Cranberry-Lemon University, Pittsburgh, PA 15213}
\affil[2]{Department of Electrical Engineering, Mount-Sheikh University, Santa Narimana, Levand}
\fi

\renewcommand{\headeright}{Technical Report}

\usepackage[dvipsnames]{xcolor}
\usepackage{lstbayes}
\usepackage{listings}
\usepackage{textcomp}
\definecolor{codegrey}{RGB}{245,245,245}
\lstdefinestyle{lfonts}{
  basicstyle   = \footnotesize\ttfamily,
  stringstyle  = \color{purple},
  keywordstyle = \color{blue!60!black}\bfseries,
}
\lstdefinestyle{lnumbers}{
  numbers     = none,
  numberstyle = \tiny,
  numbersep   = 1em,
  firstnumber = 1,
  stepnumber  = 1,
}
\lstdefinestyle{llayout}{
  breaklines       = true,
  tabsize          = 2,
  columns          = fullflexible,
}
\lstdefinestyle{lgeometry}{
  xleftmargin      = 0pt,
  xrightmargin     = 0pt,
  frame            = tb,
  framesep         = \fboxsep,
  framexleftmargin = 0pt,
}
\lstdefinestyle{lothers}{
	showstringspaces=false,
}
\lstdefinestyle{lgeneral}{
  style = lfonts,
  style = lnumbers,
  style = llayout,
  style = lgeometry,
  style = lothers,
}

\usepackage[ruled,linesnumbered]{algorithm2e}

\newtheorem{theorem}{Theorem}
\newtheorem{lemma}[theorem]{Lemma}

\begin{document}
\maketitle

\begin{abstract}
In Bayesian regression models with categorical predictors, constraints are needed to ensure identifiability when using all $K$ levels of a factor. The sum-to-zero constraint is particularly useful as it allows coefficients to represent deviations from the population average. However, implementing such constraints in Bayesian settings is challenging, especially when assigning appropriate priors that respect these constraints and general principles. Here we develop a multivariate normal prior family that satisfies arbitrary linear constraints while preserving the local adaptivity properties of shrinkage priors, with an efficient implementation algorithm for probabilistic programming languages. Our approach applies broadly to various shrinkage frameworks including Bayesian Ridge, horseshoe priors and their variants, demonstrating excellent performance in simulation studies. The covariance structure we derive generalizes beyond regression models to any Bayesian analysis requiring linear constraints on parameters, providing practitioners with a principled approach to parameter identification while maintaining proper uncertainty quantification and interpretability.
\end{abstract}


\section{Introduction}

In regression models with categorical predictors, parameter identifiability issues arise when all K levels of a factor are included. Consider a simple example where gender is included in a regression model: $$y_i = \beta_0 + \beta_{M} \text{Male}_i + \beta_{F}\text{Female}_i + ... + \epsilon_i.$$ Without additional constraints, this model is overparameterized since knowing one gender category perfectly predicts the other, leading to multicollinearity. Traditional approaches like treatment coding (setting one level as reference) work but have interpretability disadvantages - the intercept represents the mean for the reference group rather than the population average, and coefficients represent deviations from the reference group rather than from the overall mean. This complicates interpretation, especially with complex models involving multiple categorical variables or when no natural reference category exists.

A more interpretable approach retains all K coefficients while imposing the constraint that $\sum_{k=1}^K \beta_k = 0$ (sum-to-zero constraint), allowing $\beta_0$ to represent the population average and each categorical coefficient to represent the deviation from this average. However, implementing such constraints in Bayesian settings is challenging. One method redefines the last coefficient as $\beta_K = -\sum_{k=1}^{K-1} \beta_k$, which creates undesirable asymmetry among coefficients, complicates prior specification, and leads to difficult posterior geometry. Another approach implements "soft constraints" through highly informative priors (e.g., $\sum_k \beta_k \sim \mathcal{N}(0, \varepsilon)$ with very small $\varepsilon$), but this approximation can lead to numerical instabilities. We propose that the optimal solution is to directly incorporate linear constraints into prior distributions, preserving both identifiability and interpretability while enabling proper uncertainty quantification. Our approach allows for multiple arbitrary linear constraints on parameter vectors and includes efficient sampling algorithms for implementation in probabilistic programming languages. For common zero-sum constrained shrinkage priors in statistical practice, we provide four specializations as special cases of our proposed method: Bayesian ridge, hierarchical Bayesian ridge, horseshoe prior, and horseshoe prior variants. These specializations provide practitioners with ready-to-use solutions.

\section{Linear constraints via priors}

In this section, we develop a principled approach for incorporating linear constraints into Bayesian regression models with multivariate normal priors.

Given coefficient vector $\bm{\beta} = (\beta_1,\beta_2,...,\beta_K)^\top \in \mathbb{R}^K$, such that
\begin{equation}
\beta_k \mid \lambda_k \sim \mathcal{N}\left( 0, \lambda_k^2 \right), ~ k=1,2,...,K
\end{equation}
where $\lambda_k$ can be either a fixed hyperparameter or a random variable that forms a Gaussian scale mixture.

Consider a linear constraint \(A \bm{\beta} = b\), with constraint matrix $A\in \mathbb{R}^{J \times K}$ ($1\leq J \leq K-1$) with full row rank and constant vector $b \in \mathbb{R}^{K}$. We propose the following prior:
\begin{equation}
\begin{aligned}\label{eq.multinormal_cond}
\bm{\beta} \sim \mathcal{N}\left(D A^\top \left( A D A^\top \right)^{-1}b, \; D - D A^\top \left( A D A^\top \right)^{-1} A D \right)
\end{aligned}
\end{equation}
where $D = \text{diag}\left( \lambda_1^2, \lambda_2^2,...,\lambda_K^2  \right)$.

\begin{theorem}
The above prior family satisfies the given linear constraints almost everywhere.
\end{theorem}

\begin{proof}
Consider $\bm{\beta} \sim \mathcal{N}(\bm{\mu}, D)$. Let $\bm{z}=A \bm{\beta}$ represent an auxiliary observation derived from $\bm{\beta}$. The covariance structure is given by:
\begin{equation}
\operatorname{Cov}(\bm{\beta},A\bm{\beta})=D A^\top,\quad \operatorname{Cov}(A\bm{\beta}, A\bm{\beta})=A D A^\top.
\end{equation}

Therefore, the joint distribution can be written as:
\[
\begin{bmatrix}
\bm{\beta} \\
\bm{z}
\end{bmatrix}
\sim \mathcal{N}\left(
\begin{bmatrix}
\bm{\mu} \\
A\bm{\mu}
\end{bmatrix},
\begin{bmatrix}
D & D A^\top \\
A D & A D A^\top
\end{bmatrix}
\right).
\]

By standard results for conditional multivariate normal distributions, we have:
\begin{equation}
\bm{\beta}\mid A\bm{\beta}=b \sim \mathcal{N}\left(m^*, \Sigma^*\right),
\end{equation}
where the conditional mean is:
\begin{equation}
m^* = \bm{\mu} + D A^\top (A D A^\top)^{-1} (b - A\bm{\mu})
\end{equation}
and the conditional covariance matrix is:
\begin{equation}\label{eq.cov_matrix}
\Sigma^* = D - D A^\top (A D A^\top)^{-1} A D.
\end{equation}
The invertibility of $A D A^\top$ is given by Lemma \ref{lemma.ADAinvertible}.

In the absence of specific prior information, it is often desirable to employ a neutral prior centered at zero for regularization purposes. Setting $\bm{\mu}=\mathbf{0}$ yields the form presented in Eq. \ref{eq.multinormal_cond} by letting:
\begin{equation}
\bm{\beta}^*  \sim \mathcal{N}\left( D A^\top \left( A D A^\top \right)^{-1}b, \; D - D A^\top \left( A D A^\top \right)^{-1} A D \right).
\end{equation}
\end{proof}

By adjusting the covariance structure in this manner, we effectively reduce the parameter space to values that almost surely satisfy the constraint. This approach implements hard parameter constraints solely through the prior distribution. Although this conditional method does modify the original prior distribution, we can largely preserve the desired marginal distributions by appropriately scaling the diagonal elements of the covariance matrix, minimizing distortion to the original model specification. The specific scaling factors and their effects on the resulting posterior distributions will be explored in detail in Section \ref{sec.specialization}, along with practical guidelines for implementation in various modeling contexts.

\section{Sampling algorithm}

The the covariance matrix (Eq. \ref{eq.cov_matrix}) we derived is singular (see Lemma \ref{lemma.covps}). Therefore, the multivariate normal implementations in most mathematical libraries cannot directly sample from it. This section gives a general sampling algorithm and Stan implementation.

We want to sample from
$$
\bm{\beta} \sim \mathcal{N}\Bigl(m^*,\,\Sigma^*\Bigr)
$$
with
$$
m^* = D A^\top (A D A^\top)^{-1} b,\quad \Sigma^* = D - D A^\top (A D A^\top)^{-1} A D,
$$
knowing that $\Sigma^*$ is singular with rank $K-J$, where $J=\operatorname{rowrank}(A)$, $1\leq J \leq K-1$ (see Lemma \ref{lemma.covrank}).

The key idea is to sample in the full‐rank subspace where $\Sigma^*$ is nondegenerate and then map the sample back to $\mathbb{R}^K$. One common approach is to obtain an orthonormal basis for the null space of $A$. We formulate the procedure below.

\begin{algorithm}[H]
\caption{Sample $\bm{\beta} \sim \mathcal{N}(m^*,\,\Sigma^*)$}
\KwIn{Matrix $A\in\mathbb{R}^{J\times K}$, vector $b\in\mathbb{R}^{J}$, and covariance matrix $D\in\mathbb{R}^{K\times K}$}
\KwOut{Sample $\bm{\beta}\in\mathbb{R}^K$ drawn from $\mathcal{N}\bigl(m^*,\,\Sigma^*\bigr)$}

Compute
$$
m^* = D A^\top (A D A^\top)^{-1} b,
$$
and
$$
\Sigma^* = D - D A^\top (A D A^\top)^{-1} A D.
$$

Obtain an orthonormal basis $M\in\mathbb{R}^{K\times (K-J)}$ for
$$
N(A)=\{x\in\mathbb{R}^K: Ax=0\},
$$
such that
$$
AM=0 \quad \text{and} \quad M^\top M=I_{K-J}.
$$
This can be done by performing an SVD of $A$ and taking the last $(K-J)$ columns of the right singular vectors (see Algorithm \ref{algo.orthonormal_basis} for this).

Compute
$$
\Omega = M^\top \Sigma^* M \in \mathbb{R}^{(K-J)\times (K-J)}.
$$
Obtain the Cholesky factorization (this is guaranteed by Lemma \ref{lemma.lowdimrank})
$$
\Omega = L L^\top,
$$
with $L\in\mathbb{R}^{(K-J)\times (K-J)}$.\;

Draw a sample
$$
z\sim\mathcal{N}\bigl(0,I_{K-J}\bigr).
$$

Set
$$
\bm{\beta} = m^* + M L z.
$$

\textbf{return} $\bm{\beta}$\;
\end{algorithm}

\begin{theorem}
The above procedure returns sample from Eq. \ref{eq.multinormal_cond}.
\end{theorem}

\begin{proof}
The mean is
$$
E[\bm{\beta}] = m^* + M L \mathbb{E}[z] = m^*;
$$
The covariance is
$$
\operatorname{Var}(\bm{\beta}) = M L\,\operatorname{Var}(z)\,L^\top M^\top = M L L^\top M^\top = M\,\Omega\,M^\top = \Sigma^*.
$$
\end{proof}

Therefore, the above algorithm first simplifies the sampling problem to a $K-J$ subspace where the covariance is non-degenerate, and then lifts the samples to $\mathbb{R}^K$. This achieves the construction of samples from the degenerate conditional Gaussian distribution while ensuring that the constraint $A\bm{\beta} = b$ is satisfied.

\section{Sum-to-zero specialization}
\label{sec.specialization}

In this section, we demonstrate how the covariance structure can be tailored to different modeling needs, especially sum-to-zero constraints.

Note that in GP regression, the constraint matrix $A$ only required to have full row rank. Then it is possiable to set
$A = (1,1,...,1) \in\mathbb{R}^{1\times K}$ and $b=0 \in \mathbb{R}$. Therefore, the constrain becomes
\begin{equation}
     \sum_{k=1}^K \beta_k = 0.
\end{equation}
i.e., sum-to-zero constraints.

We present four specialization of our multivariate normal structure for the commonly encountered sum-to-zero constraint in statistical practice: Bayesian ridge, herarchical Bayesian ridge, horseshoe prior, and variants of the horseshoe prior. For each prior type, we derive the explicit form of the covariance matrix, discuss the implications on mathematical details, and provide efficient computational approaches that can be directly implemented in probabilistic programming languages. These specializations provide practitioners with ready-to-use solutions for imposing sum-to-zero constraints while maintaining the desirable shrinkage properties of each prior family.

\subsection{Bayesian ridge}

Bayesian ridge regression uses a normal prior which imposes a global $\ell_2$ regularization:
\begin{equation}
\bm{\beta} \sim \mathcal{N}(\bm{0}, \bm{I})
\end{equation}
where $\bm{I}$ denotes the identity matrix. By choosing $D=\bm{I}$ in Eq. \ref{eq.multinormal_cond}, the sum-to-zero constraint can be achieved by assuming the following prior
\begin{equation}
\bm{\beta} \sim \mathcal{N}\left( \bm{0},~ \sigma^2\left( \bm{I}-\frac{1}{K}\bm{1}\bm{1}^\intercal \right) \right).
\end{equation}
$\sigma^2 = \frac{K}{K-1}$ is chosen to ensure the diagonal element of covariance matrix to be $\Sigma_{k,k} = \sigma^2 \left( 1 - \frac{1}{N} \right) = 1.$ So the marginal distribution are therefore preserved:
\begin{equation}
\beta_k \sim \mathcal{N}(0, 1), k=1,2,...,K.
\end{equation}

\begin{algorithm}[H]\label{algo.sumzeroridge}
\caption{Sample from sum-to-zero Bayesian ridge}
\KwIn{A vector $x\in\mathbb{R}^{K-1}$ drawn from the standard normal distribution $\mathcal{N}(0,I_{K-1})$}
\KwOut{A vector $\bm{\beta}$ s.t. margianl normal and $1^\top \bm{\beta} = 0$. }

Compute
$$
\Sigma = I_K - \frac{1}{K}\bm{1}\bm{1}^\intercal
$$
then scale the covariance matrix as
$$
\Sigma  = \frac{K}{K-1}~\Sigma \;.
$$

Obtain the mapping matrix $M\in\mathbb{R}^{K\times (K -1)}$ from Algorithm \ref{algo.fastsumzerobasis}

Compute the matrix
$$
S = M^\top \Sigma  M \in \mathbb{R}^{(K-1)\times (K-1)}\;.
$$
Obtain the Cholesky factorization:
$$
L = \text{cholesky\_decompose}(S)\;.
$$
Compute
$$
\bm{\beta}  = M\,(L\, x)\;.
$$
\Return $\bm{\beta} $\;
\end{algorithm}

\subsection{Hierarchical Bayesian ridge}

It's possiable to introducing an additional scale factor controling global shrinkage level.
\begin{equation}
\bm{\beta} \sim \mathcal{N}\left( \bm{0}, \lambda^2 \bm{I} \right)
\end{equation}
The samller the $\lambda$, the stronger the $\ell_2$ penalty is. It is also a common choice to introduce partial pooling for variance parameters allowing adaptation to variability from the data
\begin{equation}
\lambda \sim \text{Cauchy}^+(0,1).
\end{equation}
or any other distribution with support $\mathbb{R}^+$.

By choosing $D=\lambda^2\bm{I}$ in Eq. \ref{eq.multinormal_cond}, the sum-to-zero constraint can be achieved by assuming the following prior
\begin{equation}
\bm{\beta} \sim \mathcal{N}\left( \bm{0},~ \lambda^{2}\sigma^2\left( \bm{I}-\frac{1}{K}\bm{1}\bm{1}^\intercal \right) \right).
\end{equation}
$\sigma^2 = \frac{N}{N-1}$ is chosen so the marginal distribution are preserved:
\begin{equation}
\beta_k \sim \mathcal{N}(0, \lambda^{2}), k=1,2,...,K.
\end{equation}

When sampling, only need to carry out non-centered parameterization,
\begin{equation}
\bm{\beta} = \lambda \bm{\beta^*}
\end{equation}
and then use Algorithm \ref{algo.sumzeroridge} to sample $\bm{\beta^*}$.

\subsection{Horseshoe prior}

The Horseshoe prior \citep{Carvalho2010} adopts a local-global hierarchical structure for adaptive shrinkage in regression analysis. For each regression coefficient,
\begin{equation}
\begin{aligned}
\beta_k &=  \mathcal{N}(0, \tau^2\lambda_k^2), k=1,2,...,K \\
\lambda_j &\sim C^+(0, 1) \\
\tau &\sim C^+(0, 1)
\end{aligned}
\end{equation}
The global scale parameter $\tau$ controls the overall shrinkage level, while the local scale parameter $\lambda_j$ adjusts the individual shrinkage for each coefficient.

By choosing $D=\tau^2 \text{diag}\left( \lambda_1^2, \lambda_2^2,...,\lambda_K^2  \right)$ in Eq. \ref{eq.multinormal_cond}, the sum-to-zero constraint can be achieved by assuming the following prior
\begin{align}
\pmb{\beta} \mid D \sim \mathcal{N}\left(0, \Sigma\right), ~ \Sigma=D - D \mathbf{1} \left( \mathbf{1}^\top D \mathbf{1} \right)^{-1} \mathbf{1}^\top D
\end{align}

The constraint of sum to zero introduces negative correlations among $\beta_k$, which are reflected in the off-diagonal elements of the adjusted covariance matrix $\tilde{\Sigma}$. The covariance between $\beta_k$ and $\beta_j$ is given by:
\begin{equation}
\operatorname{Cov}(\beta_k, \beta_j) = -\tau^2 \frac{\lambda_k^2 \lambda_j^2}{\sum_{l=1}^K \lambda_l^2}, \quad j \neq k
\end{equation}

The marginal variance of each $\beta_k$ is:
\begin{equation}
\operatorname{Var}(\beta_k) = \tau^2 \left( \lambda_k^2 - \frac{\lambda_k^4}{\sum_{l=1}^K \lambda_l^2} \right)
\end{equation}

Compared to the desired $\tau^2 \lambda_k^2$, this variance is slightly reduced due to the presence of the constraint. Therefore, in the constrained covariance matrix, the marginal variance could be compensated by multiplying $D$ by the factor $K / (K - 1)$, which assumes that all $\lambda_k$ are equal.

\begin{algorithm}[H]\label{algo.sumzerohorseshoe}
\caption{Sample from the sum-to-zero horseshoe prior}
\label{algo.adj_cov_cholesky}
\KwIn{A vector $x\in\mathbb{R}^{K-1}$ drawn from the standard normal distribution $\mathcal{N}(0,I_{K-1})$. Covariance matrix $D$.}
\KwOut{A vector $\bm{\beta}$ s.t. margianl horseshoe and $1^\top \bm{\beta} = 0$. }
Compute
$$
\Sigma = D - D \mathbf{1} \left( \mathbf{1}^\top D \mathbf{1} \right)^{-1} \mathbf{1}^\top D
$$
then scale the covariance matrix as
$$
\Sigma  = \frac{K}{K-1}~\Sigma \;.
$$

Obtain the mapping matrix $M\in\mathbb{R}^{K\times (K -1)}$ from Algorithm \ref{algo.fastsumzerobasis}

Compute the matrix
$$
S = M^\top \Sigma  M \in \mathbb{R}^{(K-1)\times (K-1)}\;.
$$
Obtain the Cholesky factorization:
$$
L = \text{cholesky\_decompose}(S)\;.
$$
Compute
$$
\bm{\beta}  = M\,(L\, x)\;.
$$
\Return $\bm{\beta} $\;

\end{algorithm}

It can also be used to first separate the global shrinkage parameter $\tau$ by non-central parameterization before running the algorithm.

\subsection{Horseshoe-like prior}

The proposed method is also applicable to variants of the horseshoe prior. The regularized horseshoe prior (RHS) \citep{RHS2017b} is a class of hierarchical priors designed to address sparsity in high dimension coefficients. In contrast to the traditional horseshoe prior, RHS improves model robustness by imposing moderate regularization on large coefficients through adjustments to the slab width. RHS is formulated as
\begin{equation}
\begin{aligned}
\beta_k \mid \zeta_k, c, \tau &\sim \mathcal{N}\left( 0, \tau^2 \tilde{\zeta}_k^2 \right), \quad k=1,2,...,K \\
\zeta_k &\sim \text{Student-}t^+_{\nu_1} (0,1) \\
c^2 &\sim \text{Inv-Gamma}(\nu_2/2, \nu_2s^2/2) \\
\tau &\sim \text{Student-}t^+_{\nu_3} (0, \tau_0)
\end{aligned}
\end{equation}
where $\tilde{\zeta}_k^2=\frac{c^2 \zeta_k^2}{c^2 + \tau^2 \zeta_k^2}$, and $\tau_0 = \frac{p_0}{K-p_0} \frac{\tilde\sigma}{\sqrt{n}}$. $p_0\in \left\{ 1,...,K-1 \right\}$ is a hyperparameter describing the prior belief on effective number of non-zero coefficients. $\tilde\sigma^2$ is the pseudo variance define by the likelihood and link function.

By choosing $D=\tau^2 \text{diag}\left( \tilde{\zeta}_1^2, \tilde{\zeta}_2^2,...,\tilde{\zeta}_K^2  \right)$ in Eq. \ref{eq.multinormal_cond}, the sum-to-zero constraint can be achieved by assuming the following prior
\begin{align}
\pmb{\beta} \mid D \sim \mathcal{N}\left(0, \Sigma\right), ~ \Sigma=D - D \mathbf{1} \left( \mathbf{1}^\top D \mathbf{1} \right)^{-1} \mathbf{1}^\top D.
\end{align}

When sampling, only need to carry out non-centered parameterization,
\begin{equation}
\bm{\beta} = \tau \bm{\beta^*}
\end{equation}
and then use Algorithm \ref{algo.sumzerohorseshoe} to sample $\bm{\beta^*}$.

\appendix
\renewcommand{\theequation}{\thesection.\arabic{equation}}

\section{Supplementary materials}

\subsection{Supplementary proofs}

\begin{lemma}\label{lemma.ADAinvertible}
The matrix $A D A^\top$ is invertible.
\end{lemma}

\begin{proof}
Since $\lambda_k>0, k=1,2,...,K$ are standard deviations of normal distributions, $D$ is invertible.

Let \( \tilde{A} := A D^{1/2} \in \mathbb{R}^{n \times K} \). Right multiply invertible matrix \( D^{1/2} \) does not change the rank of matrix \( A \):
$$
\text{rank}(\tilde{A}) = \text{rank}(A D^{1/2}) = \text{rank}(A) = J
$$
Note that
\[
A D A^\top = (A D^{1/2})(A D^{1/2})^\top = \tilde{A} \tilde{A}^\top
\]
so
\[
\text{rank}(A D A^\top) = \text{rank}(\tilde{A} \tilde{A}^\top) = \text{rank}(\tilde{A}) = \text{rank}(A) = J
\]

Therefore, \( A D A^\top \in \mathbb{R}^{J \times J} \) is of full rank and hence invertible.
\end{proof}

\begin{lemma}\label{lemma.covps}
The covariance matrix
\begin{equation}
\Sigma^* = D - D A^\top (A D A^\top)^{-1} A D.
\end{equation}
is positive semi-definite.
\end{lemma}

\begin{proof}
Observe that
$$
\Sigma^* = D - D A^\top \left( A^\top D A \right)^{-1} A^\top D = D^{1/2}\Bigl[I - D^{1/2} A^\top \left( A D A^\top \right)^{-1} A D^{1/2}\Bigr]D^{1/2}.
$$

Let
$$
P = D^{1/2} A^\top \left( A D A^\top \right)^{-1} A D^{1/2}.
$$
Since $A$ is full row rank and $D$ is a diagonal matrix with non-negative diagonal elements, it is known that $ADA^\top$ is invertible, thus $P$ is well-defined.

We show that $P$ is a projection matrix

Symmetry: $P^\top = P$.

Idempotence:
 $$
 P^2 = D^{1/2} A^\top \left( A D A^\top \right)^{-1}A D^{1/2} D^{1/2} A^\top \left( A D A^\top \right)^{-1}A D^{1/2} = P,
 $$

Since $P$ is a symmetric and idempotent projection matrix, all eigenvalues of $P$ and $I-P$ can only be 1 or 0. Hence $I-P$ is positive semidefinite.

Finally,
$$
\Sigma^* = D^{1/2} \Bigl[I-P\Bigr] D^{1/2}.
$$
Since $D^{1/2}=\operatorname{diag}\left(\lambda_1,\lambda_2,\ldots,\lambda_K\right)$ is positive definite, the product of positive definite matrix and positive semi-positive matrix is still positive semi-definite. Therefore, $\Sigma^*$ is a semi-positive definite matrix.

\end{proof}

\begin{lemma}\label{lemma.covrank}
If the row rank of $A$ is $J$, the covariance matrix $\Sigma^* $ is of rank $K-J$. Consequently, for any single linear constraint (including sum-to-zero constraint), the covariance matrix $\Sigma^* $ is of rank $K-1$
\end{lemma}

\begin{proof}
Since \( P \) is a symmetric and idempotent projection matrix, its rank equals the dimension of the subspace onto which it projects (i.e. the number of eigenvalues of 1), which is $D^{1/2}A^\top$. For $A\in\mathbb{R}^{J\times K}$, we have 
$$
\operatorname{rank}(P) = \operatorname{rank}(D^{1/2}A^\top) = \operatorname{rank}(A^{\top}) = J
$$

The matrix\(i-p\) is also a projection matrix with eigenvalues of 0 or 1. The number of its eigenvalue 0 (1) is equal to the number of eigenvalues in P with 1 (0). Since the rank of \( P \) is J, the rank of \( I - P \) is $K - J$.

Finally, since \( D^{1/2} \) is a full-rank matrix, $I-P$ and $D^{1/2}[I - P]D^{1/2}$ are congruent, so they have the same rank:
$$
\operatorname{rank}(\Sigma^*) = \operatorname{rank}\left(D^{1/2}[I - P]D^{1/2}\right) = \operatorname{rank}(I - P) = K - J.
   $$

In particular, sum-to-zero constraint corresponding to $\operatorname{rank}(A)=1$, so $\operatorname{rank}(\Sigma^*)=K-1$.

\end{proof}

\begin{lemma}\label{lemma.lowdimrank}
The transformed low-dimensional covariance matrix
$$
\Omega = M^\top \Sigma^* M \in \mathbb{R}^{(K-J)\times (K-J)}.
$$
is positive definite. So that we can preform Cholesky decomposition.
\end{lemma}

\begin{proof}

Let \( z \in \mathbb{R}^{K-J} \neq 0 \), and \( v = M z \in \mathbb{R}^K \). Since the columns of \( M \) are linearly independent and \( z \neq 0 \), we have \( v \neq 0 \) and \( A v = 0 \).

\[
\begin{aligned}
z^\top \Omega z &= z^\top M^\top \Sigma^* M z \\
&= v^\top \Sigma^* v \\
&= v^\top D v - v^\top D A^\top (A D A^\top)^{-1} A D v.
\end{aligned}
\]

Denote $w = D^{1/2} v$, $D$ is a symmetric positive definite matrix.

\[
\begin{aligned}
v^\top D v &= w^\top w, \\
v^\top D A^\top (A D A^\top)^{-1} A D v &= w^\top P w,
\end{aligned}
\]
where \( P = D^{1/2} A^\top (A D A^\top)^{-1} A D^{1/2} \)

we note that \( P \) is a orthogonal projection matrix,
\[
P^2 = P, \quad P^\top = P.
\]
which project $\mathbb{R}^K$ to the column space of $N = D^{1/2}A^\top$.

We note that
\[
v^\top \Sigma^* v = w^\top (I - P) w.
\]
Since \( I - P \) is a projection matrix to \( \text{Im}(P)^\perp \), \( I - P \) is semi-positive:
\[
w^\top (I - P) w \geq 0.
\]

Assume there exists a non-zero \( z \) such that \( z^\top \Omega z = 0 \), that is:

$$
w^\top (I - P) w = 0.
$$

Since \( I - P \) is a projection matrix, the equation holds if and only if \( w \in \text{Im}(P) \). In this case, there exists \( u \in \mathbb{R}^m \) such that:

$$
w = D^{1/2} A^\top u.
$$

However, \( w = D^{1/2} v \), and \( v \in N(A) \), so:

$$
A v = A D^{-1/2} w = A D^{-1/2} (D^{1/2} A^\top u) = A A^\top u = 0.
$$

Since \( A \) is full row rank, \( A A^\top \) is invertible, hence \( u = 0 \), resulting in \( w = 0 \), and thus \( v = D^{-1/2} w = 0 \), which contradicts \( v = M z \neq 0 \).

Therefore, for any non-zero \( z \in \mathbb{R}^{K-J} \), \( z^\top \Omega z > 0 \), hence \( \Omega \) is a positive definite matrix.

\end{proof}

\subsection{Supplementary algorithms}

\begin{algorithm}[H]
\label{algo.orthonormal_basis}
\caption{Compute orthonormal basis for $N(A)$}
\KwIn{Matrix $A\in \mathbb{R}^{J\times K}$}
\KwOut{Matrix $M\in \mathbb{R}^{K\times (K-J)}$ satisfying $AM = 0$ and $M^\top M = I_{K-J}$}

Compute the Singular Value Decomposition (SVD) of $A$: \\
\quad $[U,\Sigma,V]\gets \text{SVD}(A)$\;

Extract the last $(K-J)$ columns of $V$: \\
\quad $M \gets V(:,J+1:K)$\;

\textbf{return} $M$\;
\end{algorithm}

\begin{algorithm}[H]
\label{algo.fastsumzerobasis}
\caption{Compute orthonormal basis for $N(A), A=(1,1,...,1)\in \mathbb{R}^K$}
\KwIn{Positive integer $K$, with $A=(1,1,\ldots,1)\in\mathbb{R}^{1\times K}$}
\KwOut{$M\in\mathbb{R}^{K\times (K-1)}$ satisfying $AM=0$ and $M^\top M=I_{K-1}$}

\BlankLine
Initialize $M$ as a $K\times (K-1)$ zero matrix\;

\BlankLine
\For{$i\gets 1$ \KwTo $K-1$}{
    \For{$j\gets 1$ \KwTo $i$}{
        Set $M(j,i) \gets \frac{1}{\sqrt{i(i+1)}}$\;
    }
    Set $M(i+1,i) \gets -\frac{i}{\sqrt{i(i+1)}}$\;
}
\BlankLine
\textbf{return} $M$\;
\end{algorithm}

\bibliographystyle{unsrtnat}
\bibliography{references}

\begin{thebibliography}{2}
\providecommand{\natexlab}[1]{#1}
\providecommand{\url}[1]{\texttt{#1}}
\expandafter\ifx\csname urlstyle\endcsname\relax
  \providecommand{\doi}[1]{doi: #1}\else
  \providecommand{\doi}{doi: \begingroup \urlstyle{rm}\Url}\fi

\bibitem[Carvalho et~al.()Carvalho, Polson, and Scott]{Carvalho2010}
C.~M. Carvalho, N.~G. Polson, and J.~G. Scott.
\newblock The horseshoe estimator for sparse signals.
\newblock 97\penalty0 (2):\penalty0 465--480.
\newblock ISSN 0006-3444, 1464-3510.
\newblock \doi{10.1093/biomet/asq017}.

\bibitem[Piironen and Vehtari()]{RHS2017b}
Juho Piironen and Aki Vehtari.
\newblock Sparsity information and regularization in the horseshoe and other shrinkage priors.
\newblock 11\penalty0 (2).
\newblock ISSN 1935-7524.
\newblock \doi{10.1214/17-EJS1337SI}.

\end{thebibliography}

\end{document}